\documentclass[a4paper,UKenglish,cleveref,autoref,thm-restate]{lipics-v2021}
\pdfoutput=1
\usepackage{microtype}
\usepackage{algorithm}
\usepackage{algpseudocode}
\usepackage{tikz}
\usepackage{textcomp}
\usepackage{rotating}
\usepackage{wrapfig}
\usepackage{amssymb}
\usetikzlibrary{decorations.markings, arrows}

\usepackage{svg}

\title{On the Complexity of Half-Guarding Monotone Polygons}

\author{Hannah Miller Hillberg}{University of Wisconsin - Oshkosh, United States \and \url{https://uwosh.edu/cs/faculty-and-staff/hillberg/} }{hillbergh@uwosh.edu}{}{}

\author{Erik Krohn}{University of Wisconsin - Oshkosh, United States \and \url{https://faculty.cs.uwosh.edu/faculty/krohn/} }{krohne@uwosh.edu}{https://orcid.org/0000-0002-5832-8135
}{}

\author{Alex Pahlow}{University of Wisconsin - Oshkosh, United States}{pahloa45@uwosh.edu}{}{}

\authorrunning{H. M. Hillberg, E. Krohn and A. Pahlow} 

\Copyright{Hannah Miller Hillberg, Erik Krohn \and Alex Pahlow} 



\nolinenumbers

\hideLIPIcs  

\begin{document}

\maketitle

\begin{abstract}
	We consider a variant of the art gallery problem where all guards are limited to seeing to the right inside a monotone polygon. We call such guards: half-guards. We provide a polynomial-time approximation for point guarding the entire monotone polygon. We improve the best known approximation of $40$ from \cite{GKR17}, to $8$. We also provide an NP-hardness reduction for point guarding a monotone polygon with half-guards.
\end{abstract}

\section{Introduction}
An instance of the original \textit{art gallery problem} takes as input a simple polygon $P$. A polygon $P$ is defined by a set of points $V = \{v_1, v_2, \ldots, v_n\}$. There are edges connecting $(v_i, v_{i+1})$ where $i = 1, 2, \ldots, n-1$. There is also an edge connecting $(v_1, v_n)$. If these edges do not intersect other than at adjacent points in $V$ (or at $v_1$ and $v_n$), then $P$ is called a simple polygon. The edges of a simple polygon give us two regions: inside the polygon and outside the polygon. For any two points $p, q \in P$, we say that $p$ sees $q$ if the line segment $\overline{pq}$ does not go outside of $P$. The art gallery problem seeks to find a guarding set of points $G \subseteq P$ such that every point $p \in P$ is seen by a point in $G$. In the point guarding problem, guards can be placed anywhere inside of $P$. In the vertex guarding problem, guards are only allowed to be placed at points in $V$. The optimization problem is defined as finding the smallest such $G$.

\subsection{Previous Work}
There are many results about guarding art galleries. Several results related to hardness and approximations can be found in \cite{Aggarwal84,EfratH06,Eidenbenz98,Ghosh97,LL86,O87}. Whether a polynomial time constant factor approximation algorithm can be obtained for vertex guarding a simple polygon is a longstanding and well-known open problem, although a claim for one was made in \cite{BGP17}.

\noindent \textbf{Additional Polygon Structure.} Due to the inherent difficulty in fully understanding the art gallery problem for simple polygons, much work has been done guarding polygons with additional structure, see \cite{BGR17,BGR15,GKR17,NK13}. In this paper we consider monotone polygons.

\noindent \textbf{$\alpha$-Floodlights.} Motivated by the fact that many cameras and other sensors often cannot sense in 360\textdegree, previous works have considered the problem when guards have a fixed sensing angle $\alpha$ for some $0 < \alpha \leq 360$. This problem is often referred to as the \textit{$\alpha$-floodlight problem}. Some of the work on this problem has involved proving necessary and sufficient bounds on the number of $\alpha$-floodlights required to guard (or illuminate) an $n$ vertex simple polygon $P$, where floodlights are anchored at vertices in $P$ and no vertex is assigned more than one floodlight, see for example \cite{Speckmann2005,TOTH2002185}. Computing a minimum cardinality set of $\alpha$-floodlights to illuminate a simple polygon $P$ is APX-hard for both point guarding and vertex guarding \cite{AbdelkaderSHM15}. More specifically, 180\textdegree-floodlights, or \textit{half-guards}, see only in one direction. Half-guarding may have the ability to help with full-guarding. A \textit{full-guard} can see 360\textdegree. In \cite{ElbassioniKMMS11, Krohn2007}, the authors use half-guarding to show a 4-approximation for terrain guarding using full-guards. A constant factor approximation for half-guarding a monotone polygon was first shown in \cite{GKR17} and NP-hardness for vertex guarding a monotone polygon was shown in \cite{GKR18}.

\subsection{Definitions}
\noindent A simple polygon $P$ is \emph{$x$-monotone} (or simply \emph{monotone}) if any vertical line intersects the boundary of $P$ in at most two points. In this paper, we define \textit{half-guards} as guards that can see only to the right. Therefore, we redefine \emph{sees} as: a point $p$ sees a point $q$ if the line segment $\overline{pq}$ does not go outside of $P$ and $p.x \leq q.x$, where $p.x$ denotes the $x$-coordinate of a point $p$. In a monotone polygon $P$, let $l$ and $r$ denote the leftmost and rightmost point of $P$ respectively. Consider the ``top half'' of the boundary of $P$ by walking along the boundary clockwise from $l$ to $r$. We call this the \emph{ceiling} of $P$. We obtain the \emph{floor} of $P$ by walking counterclockwise along the boundary from $l$ to $r$.

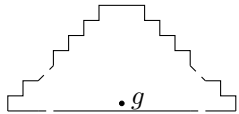
\begin{figure}[ht]
\centering
\begin{tikzpicture}[scale=0.2]
\draw (0,0) -- (0,1) -- (1,1) -- (1,2) -- (2,2);
\draw[dashed] (2,2) -- (3,3);
\draw (3,3) -- (3,4) -- (4,4) -- (4,5) -- (5,5) -- (5,6) -- (6,6) -- (6,7) -- (9,7) -- (9,6) -- (10,6) -- (10,5) -- (11,5) -- (11,4) -- (12,4) -- (12,3);
\draw[dashed] (12,3) -- (13,2);
\draw (13,2) -- (14,2) -- (14,1) -- (15,1) -- (15,0);
\draw (0,0) -- (2,0);
\draw[dashed] (2,0) -- (3,0);
\draw (3,0) -- (12,0);
\draw[dashed] (12,0) -- (13,0);
\draw (13,0) -- (15,0);

\draw[fill](7.5,0.5) circle [radius=0.15];

\node[right] at (7.5,0.7) {$g$};

\end{tikzpicture}
\caption{\centering A regular guard can see this entire monotone polygon, but needs $\Omega(n)$ half-guards.}
\label{fig:comparison-guarding}
\end{figure}

\subsection{Our Contribution}
\noindent Krohn and Nilsson \cite{NK13} give a constant factor approximation for point guarding a monotone polygon using full-guards. There are monotone polygons $P$ that can be completely guarded with one full-guard that require $\Omega(n)$ half-guards considered in this paper, see Figure \ref{fig:comparison-guarding}. Due to the restricted nature of half-guards, new observations are needed to obtain the approximation given in this paper. A $40$-approximation for this problem was presented in \cite{GKR17}. The algorithm in \cite{GKR17} places guards in 5 steps: guard the ceiling vertices, then the floor vertices, then the entire ceiling boundary, then the entire floor boundary, and finally any missing portions of the interior. We propose a modified algorithm that requires only 3 steps: guarding the entire ceiling, then the entire floor, and lastly the interior. By modifying the algorithm and providing improved analysis, we obtain an $8$-approximation. 

In addition, we show that point guarding a monotone polygon with half-guards is NP-hard. An NP-hardness proof for vertex guarding a monotone polygon with half-guards was presented in \cite{GKR18}. However, if a guard was moved off a vertex, it would see too much of the polygon and the reduction would fail. Thus, new insights were needed for point guarding. 

The remainder of the paper is organized as follows. Section \ref{approxguarding} gives an algorithm for point guarding a monotone polygon using half-guards. Section \ref{hard} provides an NP-hardness proof for point guarding a monotone polygon using half-guards. Finally, Section \ref{conclusion} gives a conclusion and possible future work.

\section{8-approximation for Point Guarding a Monotone Polygon with Half-Guards}\label{approxguarding}

We start by giving an algorithm for point guarding the boundary of a monotone polygon $P$ with half-guards. We first give a $2$-approximation algorithm for guarding the entire ceiling. A symmetric algorithm works for guarding the entire floor giving us a $4$-approximation for guarding the entire boundary of the polygon. Finally, even though the entire boundary is seen, portions of the interior may be unseen. We show that by doubling the number of guards, we can guarantee the entire polygon is seen giving an $8$-approximation.

Before we describe the algorithm, we provide some additional notation. A vertical line that goes through a point $p$ is denoted $l_p$. Given two points $p,q$ in $P$ such that $p.x < q.x$, we use $(p,q)$ to denote the points $s$ such that $p.x < s.x < q.x$. Similarly, we use $(p,q]$ to denote points $s$ such that $p.x < s.x \leq q.x$.

We first give a high level overview of the algorithm for guarding the entire ceiling boundary. Any feasible solution must place a guard at the leftmost vertex where the ceiling and floor come together (or this vertex would not be seen). We begin by placing a guard here. We iteratively place guards from left to right, letting $S$ denote the guards the algorithm has already placed. 
When placing the next guard, we let $p$ denote the rightmost point on the ceiling such that the entire ceiling from $[l,p]$ is seen by some guard in $S$. In other words, the point on the ceiling to the right of $p$ is not seen by any guard in $S$. Note that $p$ may be a ceiling vertex or any point on the ceiling. The next guard $g$ that is placed will lie somewhere on the line $l_p$. We initially place $g$ at the intersection of $l_p$ and the floor, and we slide $g$ upwards vertically along $l_p$. The algorithm locks in a final position for the guard $g$ by sliding it upwards along $l_p$ until moving it any higher would cause $g$ to no longer see some unseen point on the ceiling (i.e. a point not seen by some $g' \in S$); let $r$ be the first such point. See, for example, Figure \ref{fig:contiguous_line}. In this figure, when $g$ is initially placed on the floor, it does not see $r$, but as we slide $g$ up the line $l_p$, $r$ becomes a new point that $g$ can see. If we slide $g$ up any higher than as depicted in the figure, then $g$ would no longer see $r$, and therefore we lock $g$ in that position. We then add $g$ to $S$, and we repeat this procedure until the entire ceiling is guarded. The ceiling guarding algorithm is formally described in Algorithm \ref{alg:CeilingGuard}.

\begin{algorithm}
\caption{Ceiling Guard}
\label{alg:CeilingGuard}
\begin{algorithmic}[1]
	\Procedure{Ceiling Guard}{monotone polygon $P$}
    	\State $S \gets \{g\}$ such that $g$ is placed at the leftmost point $l$.
    	\While{there is a point on the ceiling that is not seen by a guard in $S$}
        	\State Let $p$ be the rightmost ceiling point such that the entire ceiling boundary $[l,p]$ is seen by some guard in $S$. Place a guard $g$ where $l_p$ intersects the floor and slide $g$ up. Let $r$ be the first ceiling point not seen by any guard in $S$ that $g$ would stop seeing if $g$ moved any further up. Place $g$ at the highest location on $l_p$ such that $g$ sees $r$.
            \State $S \gets S \cup \{g\}$.
            \label{step:guardPlace}
        \EndWhile
        
        \State \Return $S$
    \EndProcedure
\end{algorithmic}
\end{algorithm}

\begin{figure}[ht]
    \centering
    \includegraphics[scale=0.23]{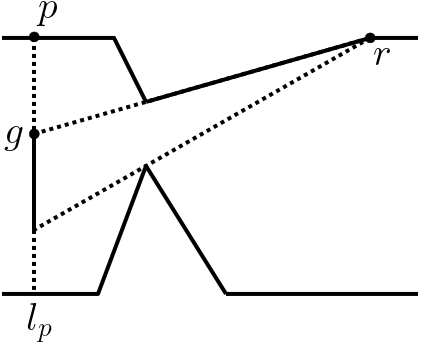}
    \caption{A guard $g$ slides up $l_p$ and sees a point $r$. If $g$ goes any higher, it will stop seeing $r$.}
    \label{fig:contiguous_line}
\end{figure}

\subsection{Sliding Analysis}
All steps, except the sliding step, can be trivially done in polynomial time. The analysis of \cite{GKR17} uses a similar sliding step but only considers guarding the vertices. When considering an infinite number of points on the ceiling, it is not immediately clear that the sliding can be done in polynomial time since each time a guard moves an $\epsilon$ amount upwards, it will see a different part of the boundary. We prove that there are at most $O(n^3)$ potential guard locations on $l_p$ that must be considered.

\begin{lemma}\label{floorBlocker} Consider a point $g$ and a point $f$ such that $g.x < f.x$. If $g$ sees $f$, then the floor (resp. ceiling) cannot block $g$ from seeing any ceiling (resp. floor) point in $(g,f)$. \end{lemma}

\begin{proof}
Assume that the polygon is monotone and $g$ sees $f$. Consider a point on the ceiling $p$ such that $g.x < p.x < f.x$. If $g$ is being blocked from seeing $p$ because of a floor vertex, then $g$ cannot see any point to the right of $p$, see Figure \ref{fig:floor_blocks_g}(left). In order for the floor to block $g$ from $p$, the floor must pierce the $\overline{gp}$ line segment. If the polygon is monotone, then the floor pierces the $\overline{gf}$ line segment, a contradiction that $g$ sees $f$. If the floor pierces the $\overline{gp}$ line segment and $g$ sees $f$, then the polygon is not monotone, see Figure \ref{fig:floor_blocks_g}(right).
\end{proof}

\begin{corollary}\label{cor:noRight}
If a point $g$ is blocked by the floor (resp. ceiling) from seeing a point $p$ on the ceiling (resp. floor), then $g$ does not see any points to the right of $p$.
\end{corollary}

\begin{figure}[ht]
    \centering
    \includegraphics[scale=0.14]{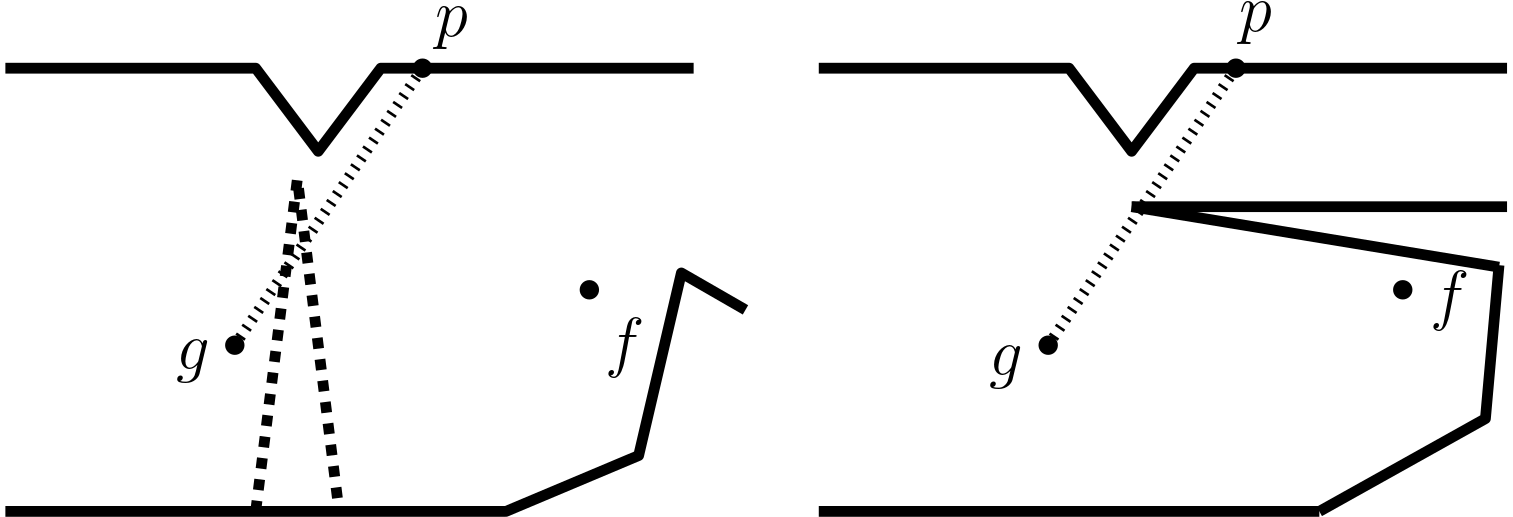}
    \caption{Assume $g.x<p.x<f.x$ and that $p$ is on the ceiling. (Left) If the floor blocks $g$ from a ceiling point $p$ and the polygon is monotone, then $g$ will not see any point $f$. (Right) If $g$ sees the point $f$ and the floor blocks $g$ from $p$, then the polygon is not monotone.}
    \label{fig:floor_blocks_g}
\end{figure}

To prove the $O(n^3)$ bound on potential guard locations, there are several cases to consider on whether or not $g$ is moved upwards and how far it needs to move. These cases inform the possible guard locations on $l_p$.

\begin{figure}[ht]
    \centering
    \includegraphics[scale=0.34]{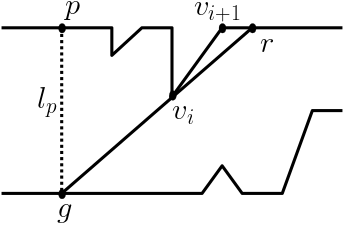}
    \caption{Ceiling visibility is maximal when $g$ is on the floor. If $g$ is moved up, then $r$ would not be seen.}
    \label{fig:gOnFloor}
\end{figure}

\noindent \textbf{Case 1:} There exists some vertex $v_i$ such that $g$ sees vertex $v_i$, but does not see vertex $v_{i+1}$ because $v_i$ is blocking $g$ from seeing $v_{i+1}$, see Figure \ref{fig:gOnFloor}. If there are multiple $v_i$ candidates, we choose the $v_i$ that is leftmost. Shoot a ray from $g$ through $v_i$ and let $r$ be the point on the boundary that is hit. The location of $r$ determines what the algorithm does.

\noindent \textbf{Case 1a:} If $r$ is on the ceiling, then no guard to the left of $g$ is able to see $r$. If $g$ is slid up, then $g$ would no longer see $r$. 

\noindent \textbf{Case 1b:} If $r$ is on the floor, then moving $g$ upwards will not cause $g$ to see any more ceiling points to the right of $r$ because $v_i$ is blocking $g$ from seeing them, see Figure \ref{fig:gAndrOnFloor}. By Lemma \ref{floorBlocker}, the floor cannot be blocking $g$ from seeing any ceiling points to the left of $r$. Therefore, moving $g$ upwards will not result in seeing any more ceiling points since the floor cannot be blocking $g$ from any point on the ceiling to the left of $r$. However, the approximation analysis relies on $g$ being as high as possible on $l_p$. Consider guards $S$ placed by the algorithm up to this point. Now consider ceiling points $R$ that are not seen by any guard $g' \in S$. We slide $g$ upwards until it would have stopped seeing some point $r' \in R$. We set $r = r'$, see Figure \ref{fig:gAndrOnFloor}. In the case where any point on $l_p$ sees all of the unseen ceiling points of $[p,v_i]$, we place $g$ on the ceiling at point $p$ and we assign $r$ to be a ceiling point directly to the right of $p$.

\begin{figure}[ht]
    \centering
    \includegraphics[scale=0.31]{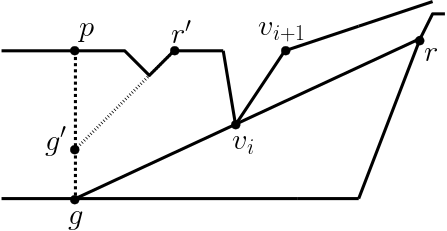}
    \caption{Ceiling visibility is maximal when $g$ is on the floor. The guard $g$ moves up until it stops seeing $r'$.}
    \label{fig:gAndrOnFloor}
\end{figure}

\noindent \textbf{Case 2:} If, for all $v_i$ that are seen by $g$, $g$ is not blocked from seeing $v_{i+1}$ by $v_i$, we slide $g$ upwards until $g$ would stop seeing an unseen ceiling point. In Figure \ref{fig:see_partial_edge}, a portion of the edge $e = [v_{i+1}, v_{i+2}]$, namely $[v_{i+1},r)$, is seen by a previously placed guard $g'$. Using Case 1, the guard $g''$ is placed to see the ceiling between $(a,p]$. If $g''$ is placed any higher, it would stop seeing $b$. Consider all such $r$ points on edges that are partially seen. Shoot a ray from $r$ through all vertices until it hits $l_p$. These locations on $l_p$ are the $G_r$ potential guard locations. In Figure \ref{fig:see_partial_edge}, the guard $g$ is placed at $g_2 \in G_r$ because it would miss the points on $e$ to the right of $r$ if moved above $g_2$.

\begin{figure}[ht]
    \centering
    \includegraphics[scale=0.15]{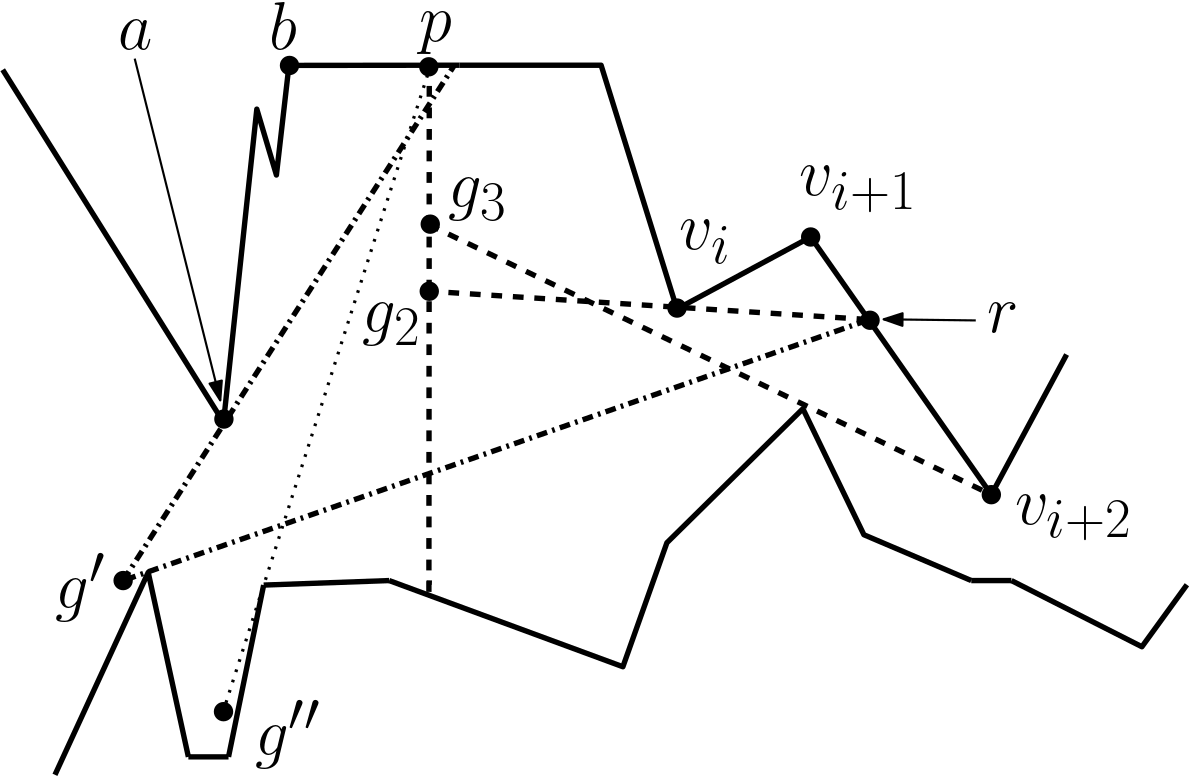}
    \caption{No point on $l_p$ sees all of $[r, v_{i+2}]$. The algorithm places a guard at $g_2$ to ensure points on the ceiling from $r$ and to its right are seen. If only case 1 locations are considered, $g_3$ would be a potential guard location considered and a portion of the edge to the right of $r$ would be unseen.}
    \label{fig:see_partial_edge}
\end{figure}

The points on $l_p$ that the algorithm must consider for Cases 1 and 2 are given by the following locations on $l_p$: 

\noindent \textbf{Potential Guard Locations A:} Rays shot from a vertex through another vertex until it hits $l_p$. Note that this includes the rays extending each edge of the polygon. There are $n$ vertices giving a total of $O(n^2)$ potential guard locations on $l_p$.

\noindent \textbf{Potential Guard Locations B:} If a guard $g$ is placed anywhere in the polygon in the range of $[v_i, v_{i+1})$, where $v_i$ and $v_{i+1}$ are ceiling vertices, then $g$ must see vertex $v_{i+1}$. If it does not, then a floor vertex must have blocked $g$ from seeing $v_{i+1}$. However, visibility of $g$ to the ceiling is not lost if $g$ is pushed up such that it sees over the floor blocker and sees $v_{i+1}$. When $g$ is placed, the subsequent guard will be placed at or beyond $l_{v_{i+1}}$. Therefore, when a guard is being placed, at most $n$ guards have been placed already. Let $C(S)$ be the ceiling locations that are hit when one shoots a ray from each previously placed guard $g'\in S$ through every vertex. In Figure \ref{fig:see_partial_edge}, a ray shot from $g'$ through the nearby floor vertex hits the ceiling at location $r \in C(S)$. With $n$ vertices, each previously placed guard can contribute at most $n$ ceiling points to $C(S)$. A ray shot from every guard through every vertex gives an upper bound of: $\vert C(S) \vert = O(n^2)$. Finally, shoot a ray from all $r \in C(S)$ through all vertices until the ray hits $l_p$. For example in Figure \ref{fig:see_partial_edge}, a ray from $r$ through vertex $v_i$ intersects $l_p$ at $g_2$. 
With $\vert C(S) \vert = O(n^2)$ points, a ray shot from each through every vertex yields a total of $O(n^3)$ rays that may hit $l_p$.

\noindent \textbf{Potential Guard Locations C:} The two points where $l_p$ hits the ceiling and the floor.

Together, a trivial analysis of the number of potential guard locations (A), (B) and (C) could generate gives an polynomial upper bound on the number of locations the algorithm has to consider at $O(n^2) + O(n^3) + 2 = O(n^3)$. The following lemma is used to show sufficiency:

\begin{lemma}\label{noSeeMiddle} Consider points $x, y$ on a vertical line such that $x$ is strictly below (resp. above) $y$ and there is a point $q$ that is strictly to the right of the vertical line. If the floor (resp. ceiling) blocks $y$ from seeing $q$, then $x$ does not see $q$. \end{lemma}

\begin{proof}
Assume that the polygon is monotone and $x$ sees $q$. The floor (resp. ceiling) cannot go through the line segment connecting $x$ and $q$. In order for the floor to block $y$ from $q$, the floor must pierce the $\overline{yq}$ line segment. If the polygon is monotone, then the floor pierces the $\overline{xq}$ line segment, a contradiction that $x$ sees $q$. If the floor pierces the $\overline{yq}$ line segment and $x$ sees $q$, then the polygon is not monotone, a contradiction. See Figure \ref{fig:xBlockedFromQ}.
\end{proof}

\begin{figure}[ht]
    \centering
    \includegraphics[scale=0.13]{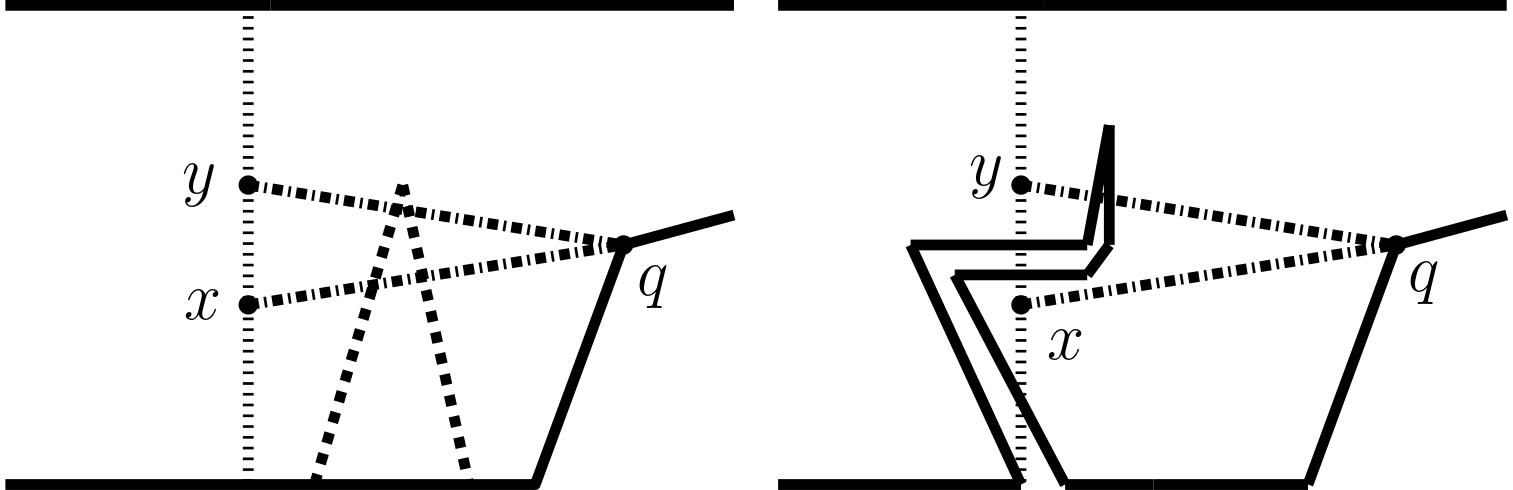}
    \caption{If $x$ and $y$ lie on a vertical line, $x$ is strictly below $y$, and $x$ sees $q$ where $x.x<q.x$, then the floor cannot block $y$ from $q$.}
    \label{fig:xBlockedFromQ}
\end{figure}

\noindent \textbf{Necessity of (A), (B) and (C):} 
If (A) guard locations are not included, then the $g''$ guard location from Figure \ref{fig:see_partial_edge} would not be considered. If one only considers locations on the $l_p$ line from (A) and (C), it is possible that a guard will be slid too far up where it misses a portion of a ceiling edge or it may not be slid far enough up. In Figure \ref{fig:see_partial_edge}, for example, the location $g_2$ would not have been considered without (B) being included. Lastly, the (C) guard locations must be included as they may be optimal, see Figure \ref{fig:floorPartC}. In the left part of the Figure, guard 1 must be on the floor. If it is moved upwards, the $[a,b]$ segment would be shifted to the right. Guard 2 would not see everything to the left of $a$ and would force a third guard to be placed much further left than guard 3. In the right part of Figure \ref{fig:floorPartC}, a guard must be placed on the ceiling to see point $a$. If guard 1 is not placed on the ceiling, then guard 2 would be further left and be unable to see point $b$.

\begin{figure}[ht]
    \centering
    \includegraphics[scale=0.15]{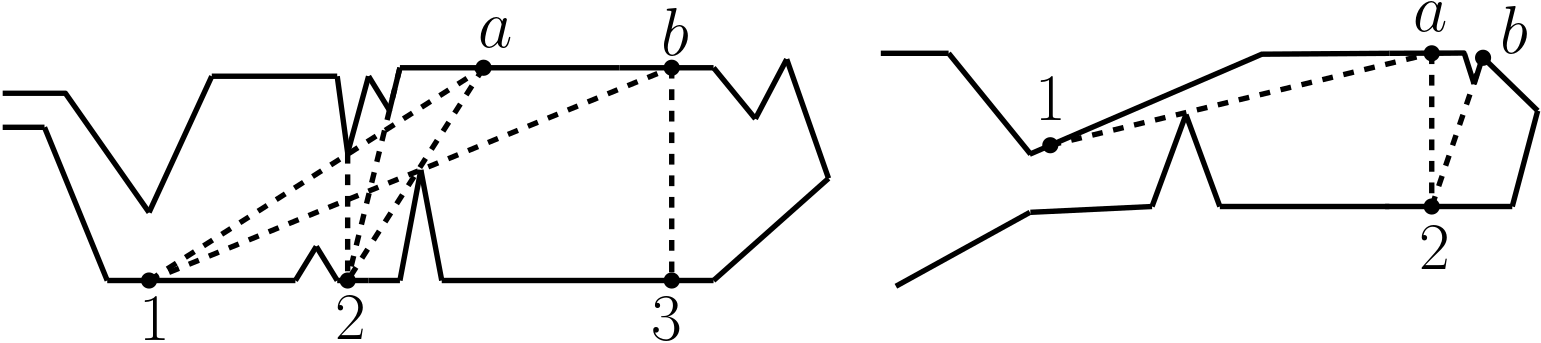}
    \caption{Necessity of guards placed on the floor and ceiling.}
    \label{fig:floorPartC}
\end{figure}

\noindent \textbf{Sufficiency of (A), (B) and (C):} Assume some set of guards $S$ has been placed. Consider 2 potential guard locations $x$ and $y$ on $l_p$ such that there are no potential guard locations between $x$ and $y$, and $x$ is below $y$. Assume that the algorithm placed a guard at location $y$. Now consider a ceiling point $q$ not seen by any guard $g \in S$. Neither $x$ nor $y$ sees $q$ but there is some location $z$ on $l_p$ that sees $q$ where $z$ is above $x$ and below $y$, see Figure \ref{fig:seesInMiddle}. In other words, $z$ is a location that the algorithm does not consider as a potential guard location. By Lemma \ref{noSeeMiddle}, the floor must block $x$ from seeing $q$ and the ceiling must block $y$ from seeing $q$. Let $v_i$ be the rightmost ceiling vertex that blocks $y$ from seeing $q$. If a ray shot from $v_{i+1}$ through $v_i$ hits the floor to the right of $l_p$, then the lowest guard location on $l_p$ to see $v_i$ would have been where the guard was placed. Such a guard location is below the $y$ location, see Figure \ref{fig:seesInMiddle} (left), a contradiction that the algorithm placed a guard at $y$. 

\begin{figure}[ht]
    \centering
    \includegraphics[scale=0.35]{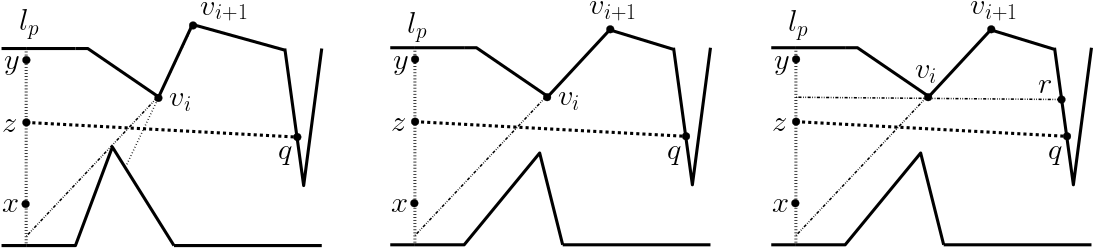}
    \caption{A guard is placed at location $y$ by the algorithm. Guard location $z$ must not exist.}
    \label{fig:seesInMiddle}
\end{figure}

If a ray shot from $v_{i+1}$ through $v_i$ goes through $l_p$, then one needs to consider multiple cases. Before considering the cases, note that this ray must intersect $l_p$ below $y$. If the ray were to have intersected above $y$, then either $y$ would have seen $v_{i+1}$, contradicting the fact that $v_i$ is the rightmost vertex that blocks $y$ from $q$. Or, if $y$ did not see $v_{i+1}$, then the floor would have blocked $y$ from $v_{i+1}$. By Corollary \ref{cor:noRight} and Lemma \ref{noSeeMiddle}, $z$ would not have seen $q$, a contradiction.

If $v_{i+1}$ was not seen by a previously placed guard, then the algorithm would have placed a guard at this intersection point because of the (A) potential guard locations, see Figure \ref{fig:seesInMiddle} (middle). This contradicts the algorithm choosing guard location $y$. If $v_{i+1}$ was seen by a previously placed guard, then consider the rightmost ceiling point $r$ between $(v_{i+1}, q)$ that was seen by a guard in $S$. Such a location must exist since some guard $g' \in S$ sees $v_{i+1}$. 
Potential guard locations (B) would have shot a ray from $r$ through $v_i$. If this ray hits $l_p$ below $y$, then the algorithm would have chosen that point as our guard location, contradicting that the algorithm chose $y$. If this ray hits $l_p$ above $y$, then $v_i$ could not have blocked $y$ from seeing $q$. Therefore, it is not possible for guard location $z$ to exist.

\subsection{Approximation Analysis}\label{approx}

We will now prove why Algorithm \ref{alg:CeilingGuard} will place no more than 2 times the number of guards in the optimal solution. The argument is similar to the argument presented in \cite{GKR17}. An optimal solution $\mathcal{O}$ is a minimum cardinality guard set such that for any point $p$ on the ceiling of $P$, there exists some $g \in \mathcal{O}$ that sees $p$. The argument will be a charging argument; every guard placed will be charged to a guard in $\mathcal{O}$ in a manner such that each guard in $\mathcal{O}$ will be charged at most twice. First, charge the leftmost guard placed to the optimal guard at the same location, call this guard $g_1$. All optimal guardsets include a guard at the leftmost vertex in the polygon else the leftmost vertex is unseen. Now consider two consecutive guards $g_i, g_{i+1} \in S$ returned by the algorithm. Since $g_1$ has already been charged, we must find an optimal guard to charge $g_{i+1}$ to.

\noindent \textbf{Case 1:} If at least one optimal guard is in $(g_i, g_{i+1}]$, then we charge $g_{i+1}$ to any optimal guard, chosen arbitrarily, in that range.

\noindent \textbf{Case 2:} 
If there is no optimal guard in $(g_i, g_{i+1}]$, then consider the point on the ceiling directly above $g_{i+1}$, call this point $p$. The $g_{i+1}$ guard was placed on $l_p$ because no previously placed guard saw the ceiling point to the right of $p$ and the entire ceiling from $[l,p]$ is seen guards in by $S$. Consider the guard $g_i$. If $g_i$ sees $p$, then it must be the case that the floor blocks $g_i$ from seeing to the right of $p$ (otherwise $p$ was chosen incorrectly), see Figure \ref{fig:otherGuardExists}(left). By Corollary \ref{cor:noRight}, $g_i$ does not see any ceiling point to the right of $p$. If $g_i$ does not see $p$, then some previously placed guard must have seen $p$. Consider the first floor vertex $v_k$ that blocks $g_i$ from seeing $p$. Since $p$ was seen by a previously placed guard, a ray was shot from $p$ through $v_k$ when considering where to place $g_i$. Let $g_p$ be the point on $l_{g_i}$ where the $\overrightarrow{pv_k}$ ray hits. The $g_i$ guard must be below $g_p$ otherwise $g_i$ would have seen $p$ and more importantly, the ceiling point to the right of $p$, see Figure \ref{fig:otherGuardExists}(right). In either case, $g_i$ does not see to the right of $p$. The reason that $g_i$ stopped moving upwards is because it saw some point $r$ on the ceiling that no previously placed guard saw. Since $g_i$ cannot see to the right of $p$, $r$ must be to the left of $p$. By assumption, the optimal guard $o'$ that sees $r$ must be to the left of $g_i$. If $g_i$ were to have moved any higher up, it would have missed $r$. Any guard that sees $r$ must be ``below'' the $\overline{g_ir}$ line, see the gray shaded region of Figure \ref{fig:otherGuardExists}. Any point on the ceiling that $o'$ sees to the right of $r$, $g_i$ will also see (Lemma 1, \cite{GKR17}). Therefore, $g_i$ dominates $o'$ with respect to the ceiling to the right of $r$. We charge $g_{i+1}$ to $o'$ and $o'$ cannot be charged again.

\begin{figure}[ht]
    \centering
    \includegraphics[scale=0.15]{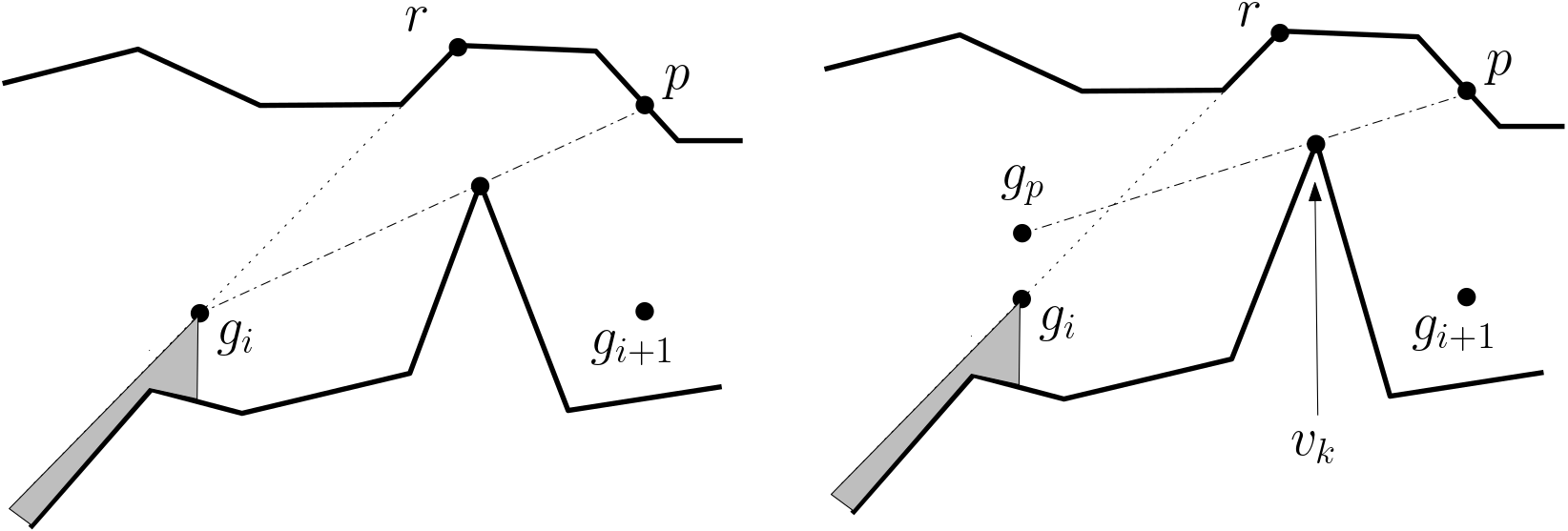}
    \caption{If no optimal guard exists in $(g_i, g_{i+1}]$, then $o'$ must exist to see $r$ such that $o'.x \leq g_i.x$.}
    \label{fig:otherGuardExists}
\end{figure}

\begin{figure}[ht]
    \centering
    \includegraphics[scale=0.27]{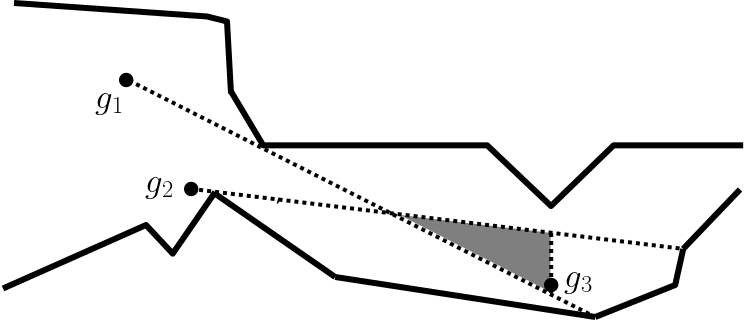}
    \caption{$g_1, g_2$ and $g_3$ see all of the boundary to the right of $g_1$. The shaded region is unguarded.}
    \label{fig:middlePocket}
\end{figure}

The entire ceiling can be guarded with at most $2\cdot\vert \mathcal{O}\vert$ guards. A similar algorithm is applied to the floor to give at most $4\cdot\vert \mathcal{O}\vert$ guards to guard the entire boundary. Finally, even though the entire boundary is guarded, it is possible that a portion of the interior is unseen, see Figure \ref{fig:middlePocket}. Let us assume that a guardset $G = \{g_1, g_2,\ldots g_k\}$ guards the entire boundary of a monotone polygon such that for all $i$, $g_i < g_{i+1}$. In \cite{KrohnN12}, they prove that between any consecutive guards of $G$, a region can exist that is unseen by any of the guards in $G$. However, they prove that the region is convex and can be guarded with 1 additional guard. If $\vert G \vert = k$, then there are at most $k-1$ guards that need to be added to guard these unseen interior regions. This doubles the approximation to give us the following theorem:

\begin{theorem}
There is a polynomial-time $8$-approximation algorithm for point guarding a monotone polygon with half-guards.
\end{theorem}

\section{NP-Hardness for Point Guarding a Monotone Polygon with Half-Guards}\label{hard}

\begin{figure}[ht]
    \centering
    \includegraphics[scale=0.24]{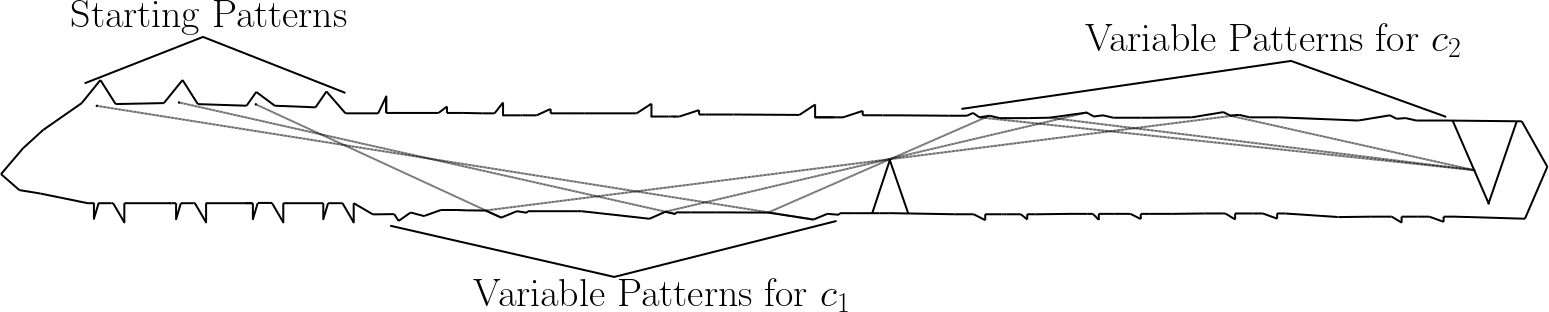}
    \caption{A high level overview of the reduction.}
    \label{fig:nphardOverview}
\end{figure}

In this section, we show that point guarding a monotone polygon with half-guards is NP-hard. NP-hardness for vertex guarding a monotone polygon with half-guards was shown in \cite{GKR18}. However, moving the guards a small amount off of the vertex causes the entire reduction to fail. Additional insight and additional gadgets/patterns were needed to show hardness for point guarding a monotone polygon with half-guards.

The reduction is from \emph{3SAT}. A 3SAT instance $(X, C)$ contains a set of Boolean variables, $X = \{x_1, x_2, \ldots, x_n\}$ and a set of clauses, $C = \{c_1, c_2, \ldots, c_m\}$. Each clause contains three literals, $c_i = (x_j \vee x_k \vee x_l)$. A 3SAT instance is satisfiable if a satisfying truth assignment for $X$ exists such that all clauses $c_i$ are true. We show that any 3SAT instance is polynomially transformable to an instance of point guarding a monotone polygon using half-guards. We construct a monotone polygon $P$ from the 3SAT instance such that $P$ is guardable by $K = n*(1 + 2m) + 1$ or fewer guards if and only if the 3SAT instance is satisfiable.

The high level overview of the reduction is that specific potential guard locations represent the truth values of the variables in the 3SAT instance. Starting patterns are placed on the ceiling on the left side of the polygon, see Figure \ref{fig:nphardOverview}. In these starting patterns, one must choose one of two guardset locations in order to guard distinguished edges for that particular pattern. A \emph{distinguished edge} is an edge that is only seen by a small number of specific guard locations. Then, to the right of the starting patterns, variable patterns are placed on the floor, then the ceiling, then the floor, and so on for as many clauses as are in the 3SAT instance. In each variable pattern, similar to a starting pattern, certain guard locations will represent a truth assignment of true/false for a variable ($x_i$). This Boolean information is then ``mirrored rightward'' such that there is a consistent choice of all true $x_i$ locations or all false $\overline{x_i}$ locations for each variable. This differs from previous results where Boolean information was mirrored from the ``left side'' of the polygon/terrain to the ``right side'' of the polygon/terrain and then back to the left side \cite{KingK11, KrohnN12}. A \emph{distinguished clause vertex} is placed to the right of each sequence of variable patterns such that only the guard locations representing the literals in the specific clause can see the distinguished clause vertex. A high level example of the entire reduction is shown in Figure \ref{fig:nphardOverview}.

\begin{figure}[ht]
    \centering
    \includegraphics[scale=0.45]{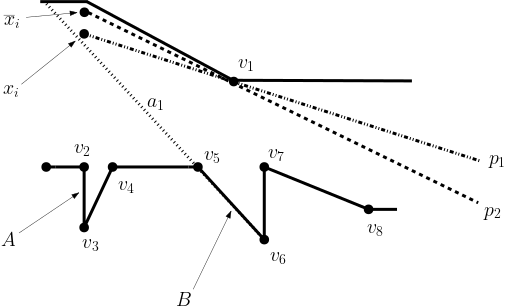}
    \caption{A starting pattern for variable $x_i$.}
    \label{fig:nphardStarting}
\end{figure}

\subsection{Hardness Details}\label{details}

\noindent \emph{Starting Pattern:} The starting pattern, shown in Figure \ref{fig:nphardStarting}, appears along the ceiling of the left side of the monotone polygon a total of $n$ times, each corresponding to a variable from the 3SAT instance, in order from left to right ($x_1, x_2, \ldots, x_n$). In each pattern, there are 2 distinguished edges with labels: $\{A = \overline{v_{2}v_{3}}, B = \overline{v_{5}v_6}\}$. These edges are seen by a specific, small range of points in each starting pattern. It is important to note that no other point in any other starting pattern sees these distinguished edges. 

In order to see $A$, a guard must be placed on the vertical line that goes through $v_2$. We can make the $A$ edge be almost vertical such that any guard placed to the left of the vertical line through $A$ will miss $v_3$. In order to guard the entire pattern with 1 guard, that guard must see $B$ as well. We are left with a small range of potential guard locations as seen in Figure \ref{fig:nphardStarting}. Any guard placed above $a_i$ on the $l_{v_2}$ line will see these distinguished edges. Foreshadowing the variable mirroring, a guard placed in the $x_i$ region will represent a true value for $x_i$. Any guard placed in the $\overline{x_i}$ region will represent a false value for $x_i$.

\begin{figure}[ht]
    \centering
    \includegraphics[scale=0.32]{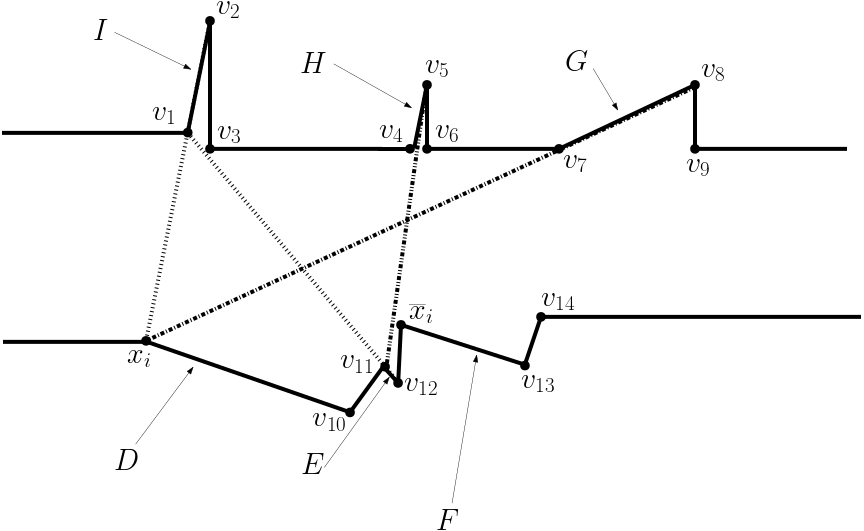}
    \caption{A variable pattern with critical lines of sight shown.}
    \label{fig:nphardVariable}
\end{figure}

\noindent \emph{Variable Pattern}: On the floor of the polygon to the right of all the $n$ starting patterns are the first $n$ variable patterns, one for each variable, that verify and propagate the assigned truth value of each variable. The variables are in reverse order from the initial starting pattern  
($x_n, x_{n-1}, \ldots, x_1$). 
When the variables are ``mirrored'' rightward again to the ceiling, the ordering will again reverse. 

A single variable pattern is shown in Figure \ref{fig:nphardVariable}. There are 6 distinguished edges located at $\{D = \overline{x_iv_{10}}, E = \overline{v_{11}v_{12}}, F = \overline{\overline{x_i}v_{13}}, G = \overline{v_{7}v_{8}}, H = \overline{v_{4}v_{5}}, I = \overline{v_{1}v_{2}}\}$. Of those, $E, G, H$ and $I$ are only visible within the variable pattern. Further, no guard that sees $I$ can also see $H$. Therefore, at least 2 guards are required to see the vertices in a variable pattern. However, with $D$ or $F$ seen by a previously placed guard, then 2 guards are sufficient for guarding the variable pattern. If neither $D$ nor $F$ is seen by a previously placed guard, then 3 guards are required to see all of $\{ D, E, F, G, H, I\}$. Assume that no guard in a previous pattern saw either $D$ or $F$ in the variable pattern. If this is the case, then a guard must be placed on (or to the left of) the $l_{x_i}$ line in order to see $D$. No guard that sees $D$ also sees $H$. If a guard that sees $D$ does not see $I$, then 3 guards are required since no guard sees both $H$ and $I$.

We assume that the guard that sees $D$ also sees $I$. The only guard that sees both $D$ and $I$ is at $x_i$. If the guard is moved left, it will miss $D$. If it is moved up, it will miss $I$. Therefore, a guard is placed at $x_i$ and we now also see $G$. This leaves $E, F$ and $H$ to be guarded. A guard that sees $E$ and $H$ must lie on $v_{11}$. If the guard is pushed higher, it will not see $H$. If it is pushed left or right, it will not see $E$. The vertex $\overline{x_i}$ blocks $v_{11}$ from seeing $F$. Therefore, the $F$ edge goes unseen and a third guard is required.

In summary, if a previously placed guard does not see $D$ nor $F$, 3 guards are required to guard the variable pattern.

Let us assume that $F$ is seen by a previously placed guard. If this happens, then $\{D, E, G, H, I\}$ must still be guarded. To see all of $D$, a guard must be placed on the vertical line through (or to the left of) $x_i$, see Figure \ref{fig:nphardVariable}. If the guard is placed above $x_i$, then $I$ would remain unguarded. If this happens, then 2 guards are still required to see $I$ and $H$. Therefore, a guard that sees $D$ must also see $I$, and the only such location is $x_i$. Placing a guard at $x_i$ sees edges $D, G$ and $I$, leaving $E$ and $H$ to be seen. The only guard that sees all of $E$ and $H$ is at vertex $v_{11}$. Therefore, two guards suffice and they are placed at $x_i$ and $v_{11}$.

Now let us assume that $D$ is seen by a previously placed guard. If this happens, then $\{E, F, G, H, I\}$ must still be guarded. In order to see $I$, a guard must be on (or to the left of) $l_{v_1}$. Edge $E$ is angled in such a way that it does not see any of the vertical line below $v_1$. 

If a guard is placed at $v_1$, then edges $E$ and $I$ are guarded and we need to place a second guard that sees $F, G$ and $H$. It should be noted that a guard placed at $v_1$ does not see to the right of $\overline{x_i}$ since $v_3$ is blocking $v_1$ from seeing too far to the right. A ray shot from $v_5$ through $v_4$ goes just above the $\overline{x_i}$ vertex, 
so placing a guard too far above $\overline{x_i}$ will not see edge $H$. Thus, a guard must be placed arbitrarily close to $\overline{x_i}$ to see $F, G$ and $H$. Therefore, two guards suffice in this instance and guards are placed arbitrarily close to both $\overline{x_i}$ and $v_1$.

It is not immediately obviously that the guard locations of $\overline{x_i}$ and $v_1$ are necessary guard locations when edge $D$ was seen by a previously placed guard. If one tries to place a guard that sees $I$ and a different subset of distinguished edges (other than $E$), one will require more than 2 guards to see the variable pattern. If a guard placed to see edge $I$ does not see edge $E$, then there are 3 regions on (or to the left of) $l_{v_1}$ where a guard could be placed. 

It is important to note that a ray shot from $v_{13}$ through $\overline{x_i}$ is above a ray shot from $v_8$ through $v_7$ when they cross $l_{v_1}$. In other words, no guard that sees $I$ can see both $F$ and $G$. The following 3 cases assume that the guard that sees $I$ does not see $E$ and it will lead to a contradiction in each case. This will show that when $D$ is seen by a previously placed guard, one must put guards arbitrarily close to both $\overline{x_i}$ and $v_1$.
\begin{enumerate}
    \item If a guard is placed on $l_{v_1}$ that sees neither $F$ nor $G$. In this case, edges $E, F, G$ and $H$ are unseen. A guard that sees $E$ and $H$ must lie on $v_{11}$. If the guard is pushed higher, it will not see $H$. If it is pushed left or right, it will not see $E$. The vertex $\overline{x_i}$ blocks $v_{11}$ from seeing $F$ and $G$. Therefore, the $F$ and $G$ edges go unseen and a third guard is required.
    
    \item If a guard placed on $l_{v_1}$ to see $I$ and also sees $F$, then edges $E, G$ and $H$ are unseen. Similar to the case above, a guard that sees $E$ and $H$ must lie on $v_{11}$. The vertex $\overline{x_i}$ blocks $v_{11}$ from seeing $G$. Therefore, the $G$ edge goes unseen and a third guard is required.

    \item If a guard placed on $l_{v_1}$ to see $I$ and also sees $G$, then edges $E, F$ and $H$ are unseen. Similar to the case above, a guard that sees $E$ and $H$ must lie on $v_{11}$. The vertex $\overline{x_i}$ blocks $v_{11}$ from seeing $F$. Therefore, the $F$ edge goes unseen and a third guard is required.

\end{enumerate}

To summarize, if $D$ or $F$ is seen, then in order to guard the entire variable pattern with two guards, there are exactly two sets of potential guard locations: $\{x_i, v_{11}\}$, and $\{v_1, \overline{x_i}\}$.

\begin{figure}[ht]
    \centering
    \includegraphics[scale=0.22]{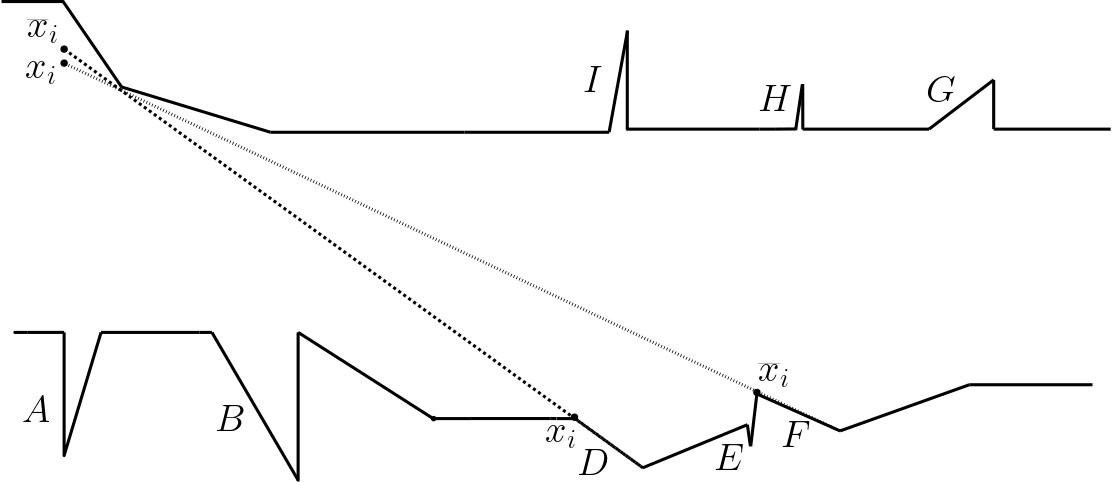}
    \caption{Starting pattern for $x_i$ interacting with first variable pattern.}
    \label{fig:nphardStartingVariableInteraction}
\end{figure}

\noindent \emph{Connecting a Starting Pattern to a Variable Pattern}: We will consider an arbitrary variable $x_i$. In the starting pattern for $x_i$, there was a line segment of potential guard locations, namely anything above the $a_1$ line segment on $l_{v_2}$, see Figure \ref{fig:nphardStarting}. The starting pattern is connected to the variable pattern in such a way that any guard placed in the $x_i$ region will see edge $F$ in the subsequent variable pattern. Any guard placed in the $\overline{x_i}$ region will see edge $D$ in the variable pattern. A guard placed in the $\overline{x_i}$ region cannot see the $F$ edge in the first variable pattern because vertex $v_1$ in the starting pattern prevents it from seeing too far right. A guard placed in the $x_i$ region cannot see the $D$ edge in the variable pattern because the angle of the $D$ edge is such that it goes above the $x_i$ region in the starting pattern. In the starting pattern between the $x_i$ and $\overline{x_i}$ region is an ``empty'' region that sees neither $D$ nor $F$ in the first variable pattern for $x_i$.

\begin{figure}[ht]
    \centering
    \includegraphics[scale=0.38]{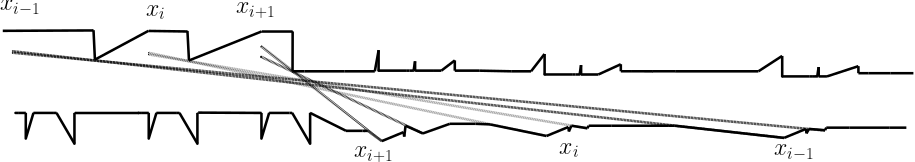}
    \caption{Multiple starting patterns interacting with multiple variable patterns.}
    \label{fig:nphardMultipleStartingVariableInteraction}
\end{figure}

\noindent \emph{Connecting Multiple Starting Patterns to Multiple Variable Patterns}: Consider three consecutive starting patterns for $x_{i-1}, x_i$ and $x_{i+1}$. As a reminder, the ordering of the variable patterns in the first grouping of variable patterns is $x_{i+1}, x_i, x_{i-1}$, see Figure \ref{fig:nphardMultipleStartingVariableInteraction}. No guard placed in the starting pattern of $x_{i-1}$ will see edges $D$ and $F$ in the variable pattern for $x_i$ since those edges are angled in such a way that they do not see to the left of the starting pattern for $x_i$. To ensure no guard placed in the $x_{i+1}$ starting pattern sees the $D$ or $F$ edges in the $x_i$ variable pattern, the $v_1$ vertex of the $x_{i+1}$ starting pattern is used to block the $x_{i+1}$ line segments from seeing right of the first $x_{i+1}$ variable pattern.

\begin{figure}[ht]
    \centering
    \includegraphics[scale=0.26]{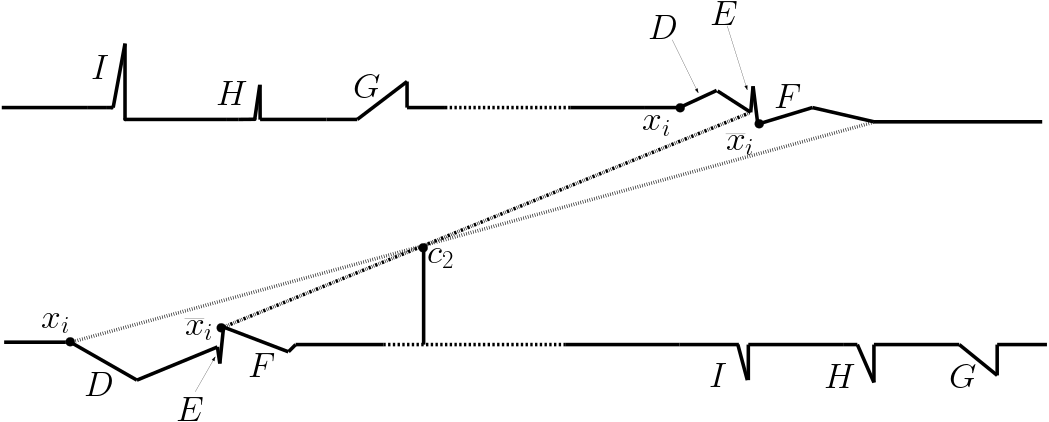}
    \caption{A variable pattern being mirrored rightward.}
    \label{fig:nphardVariableMirroring}
\end{figure}

\noindent \emph{Connecting a Variable Pattern to a Variable Pattern}: Consider the first variable pattern for $x_i$. A first variable pattern for $x_i$ is connected to the subsequent variable pattern for $x_i$ in such a way that any guard placed in the $x_i$ region will see edge $F$ in the subsequent variable pattern. Any guard placed in the $\overline{x_i}$ region will see edge $D$ in the subsequent variable pattern. See Figure \ref{fig:nphardVariableMirroring}. 
A guard placed in the $x_i$ region cannot see the $D$ edge in the next variable pattern because the angle of the $D$ edge is such that it goes to the right of the $x_i$ region in the previous variable pattern. A guard placed in the $\overline{x_i}$ region cannot see the $F$ edge in the next variable pattern because a $c_2$ vertex, which is to the right of the sequence of variable patterns, prevents it from seeing too far right. 

\begin{figure}[ht]
    \centering
    \includegraphics[scale=0.35]{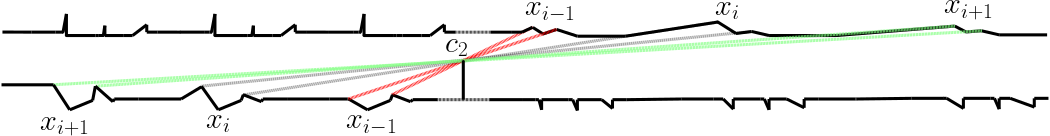}
    \caption{Multiple variable patterns being mirrored rightward.}
    \label{fig:nphardMultipleVariablePatterns}
\end{figure}

\noindent \emph{Connecting Multiple Variable Patterns to Multiple Variable Pattern}: We will consider consecutive variables $x_{i-1}, x_i$ and $x_{i+1}$, see Figure \ref{fig:nphardMultipleVariablePatterns}. To ensure guards placed in the floor $x_{i-1}$ variable pattern do not see the $D$ or $F$ edges in the subsequent ceiling variable pattern for $x_i$, the $c_2$ vertex is used to block the guards from seeing too far to the right. No guards placed in the floor variable pattern of $x_{i+1}$ will see edges $D$ and $F$ in the subsequent ceiling variable pattern for $v_i$ since the $D$ and $F$ edges are angled in such a way that they do not see to the left of the first variable pattern for $x_i$.

\noindent \emph{Clauses}: For each clause $c$ in the 3SAT instance, there is a sequence of variable patterns $x_1, \ldots, x_n$ along either the floor or ceiling of the polygon. Immediately to the right of one such sequence of variable patterns exists a clause pattern. A clause pattern consists of one vertex such that the vertex is only seen by the variable patterns corresponding to the literals in the clause; see Figure \ref{fig:nphardClauseCompleteExample}. The distinguished vertex of the clause pattern is the $c_3$ vertex. This vertex is seen only by specific vertices in its respective sequence of variable patterns.

\begin{figure}[ht]
    \centering
    \includegraphics[scale=0.32]{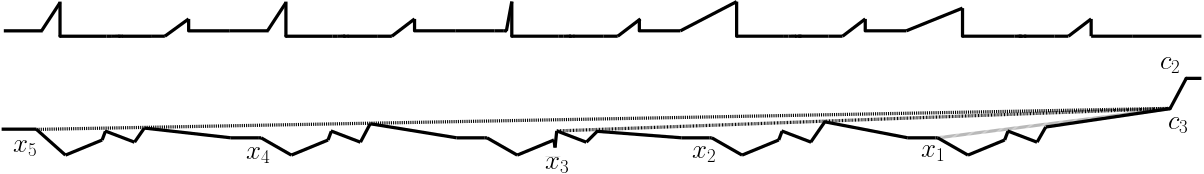}
    \caption{High level overview of a clause point $c_3$ being placed for clause $c_i = x_1 \vee \overline{x_3} \vee x_5$.}
    \label{fig:nphardClauseCompleteExample}
\end{figure}

To see how a clause distinguished point is placed in the polygon, consider Figure \ref{fig:nphardClauseCompleteExample} that represents the clause $x_1 \vee \overline{x_3} \vee x_5$. For the floor variable patterns, all $\overline{x_i}$ vertices are blocked from seeing the $c_3$ clause point by the $v_{14}$ vertex in their variable pattern. All $x_i$ points are blocked from seeing $c_3$ by the $\overline{x_i}$ vertex in their variable pattern. In the case of $\overline{x_3}$, the $v_{14}$ vertex is lowered an epsilon amount such that the $v_{14}$ vertex is no longer blocking $\overline{x_3}$ from seeing $c_3$. In the case of $x_1$ and $x_5$, consider the ray that is shot from $v_{10}$ through $x_i$ in the variable pattern. Push $x_i$ up an epsilon amount on that ray until it sees the $c_3$ distinguished point. Moving the $x_i$ vertex a small amount will force the $I$ and $G$ edges in the local variable pattern to be modified slightly. However, no global visibility is affected with this small potential movement of $x_i$. Once all $x_i$ vertices are placed for the first clause, then $c_2$ is drawn such that the visibilities moving rightward are correct. Therefore, $c_2$ still blocks each respective vertex from seeing too far to the right.

\noindent \emph{Putting it all together:} We choose our truth values for each variable in the starting variable patterns. The truth values are then mirrored in turn between variable patterns on the floor and the ceiling. Consider the example of Figure \ref{fig:nphardClauseCompleteExample} the 3SAT clause corresponds to $c_i = (x_1 \vee \overline{x_3} \vee x_5)$. Hence, a guard placement that corresponds to a truth assignment that makes $c_i$ true, will have at least one guard on or near the $x_1, \overline{x_3}$ or $x_5$ vertex and can therefore see vertex $c_3$ without additional guards. We still have variables $x_2$ and $x_4$ on the polygon, however, none of them nor their negations see the vertex $c_3$. They are simply there to transfer their truth values in case these variables are needed in later clauses.

The monotone polygon we construct consists of $11n + (16n+3)m + 2$ vertices where $n$ is the number of variables and $m$ is the number of clauses. Each starting variable pattern has 11 vertices, each variable pattern 16 vertices, the clause pattern has 3 vertices, plus 2 vertices for the leftmost and rightmost points of the polygon. Exactly $K = n*(1 + 2m) + 1$ guards are required to guard the polygon. A guard is required to see the distinguished edges of the starting patterns and 2 guards are required at every variable pattern, of which there are $(mn)$ of them. Lastly, since a starting pattern cannot begin at the leftmost point, a guard is required at the leftmost vertex of the polygon. If the 3SAT instance is satisfiable, then $K = n*(1 + 2m) + 1$ guards are placed at locations in accordance with whether the variable is true or false in each of the patterns. Each clause vertex is seen since one of the literals in the associated clause is true and the corresponding $c_3$ clause vertex is seen by some guard.

\begin{theorem}
Finding the smallest point guard cover for a monotone polygon using half-guards is NP-hard.
\end{theorem}

\section{Conclusion and Future Work}\label{conclusion}
In this paper, we present an 8-approximation for point guarding a monotone polygon with half-guards. We also show that point guarding a monotone polygon with half-guards is NP-hard. Future work might include finding a better approximation for both the point guarding and vertex guarding version of this problem. Insights provided in this paper may help with guarding polygons where the guard can choose to see either left or right, or in other natural directions. One may also be able to use these ideas when allowing guards to see $180$\textdegree\ but guards can choose their own direction ($180$\textdegree-floodlights).

\bibliography{main}

\end{document}